\documentclass[12pt]{article} 

\title{
  Stackelberg Independence%
}
\usepackage{datetime}
\newdateformat{monyear}{\monthname[\THEMONTH] \THEYEAR} 

\date{%
  \monyear\today\footnote{The latest version is available at \href{https://toomas.hinnosaar.net/stackelberg.pdf}{\url{https://toomas.hinnosaar.net/stackelberg.pdf}
  }
}
} 

\author{
Toomas Hinnosaar%
\thanks{Collegio Carlo Alberto, \href{mailto:toomas@hinnosaar.net}{\url{toomas@hinnosaar.net}}. %
}%
}

\usepackage{anysize}
\usepackage[T1]{fontenc}
\usepackage{lmodern}
\usepackage{hyperref}
\usepackage{subfig}
\usepackage{pgf,tikz}
\usepackage{natbib}
\usepackage{appendix}
\usepackage{enumerate}
\usepackage{amssymb,amsmath,amsthm}
\usepackage{thmtools,thm-restate}
\declaretheorem[name=Proposition]{proposition}  
  
\declaretheorem[name=Corollary]{corollary}  
\declaretheorem[name=Lemma]{lemma}  
\declaretheorem[name=Definition]{definition}  
  
\renewenvironment{proof}[1][Proof]{\begin{trivlist}
		\item[\hskip \labelsep {\bfseries #1}]}{\qed \end{trivlist}}

\newcommand{\N}{\mathbb N}
\newcommand{\expect}{\mathbb E}
\usepackage{cleveref}
\newcommand{\prt}{\mathcal{I}}
\newcommand{\bprt}{\boldsymbol{\prt}}
\usepackage{color}

\newcommand{\bx}{\boldsymbol{x}}
\newcommand{\bn}{\boldsymbol{n}}
\newcommand{\hn}{\hat{n}}
\newcommand{\hT}{\hat{T}}
\newcommand{\hbn}{\boldsymbol{\hn}}
\newcommand{\oX}{\overline{X}}

\crefname{figure}{figure}{figures}
\crefname{equation}{equation}{equations}

\begin{document}
\maketitle

\begin{abstract}
The standard model of sequential capacity choices is the Stackelberg quantity leadership model with linear demand. I show that under the standard assumptions, leaders' actions are informative about market conditions and independent of leaders' beliefs about the arrivals of followers. However, this Stackelberg independence property relies on all standard assumptions being satisfied. It fails to hold whenever the demand function is non-linear, marginal cost is not constant, goods are differentiated, firms are non-identical, or there are any externalities. I show that small deviations from the linear demand assumption may make the leaders' choices completely uninformative.
\end{abstract}  

\emph{JEL}:
C72, 
C73, 
D43, 
L13	

\emph{Keywords}: sequential games, oligopolies, Stackelberg leadership model

\section{Introduction} \label{S:intro}

How can one determine market characteristics and get early welfare estimates in newly developing markets? Do these markets require policy interventions? What will the long-term outcomes look like? In many markets, firms enter and build capacities sequentially. For example, ride-sharing companies (such as Uber, Lyft, BlaBlaCar, and Taxify) typically enter each geographic location at different moments in time. A natural model to study this kind of markets is the Stackelberg quantity leadership model, where firms choose their quantities (capacities) while observing the moves of earlier entrants.

In this paper, I show that under the standard assumptions of the Stackelberg model, the questions stated above are easy to answer. The standard assumptions are (1) linear demand, (2) identical firms, and (3) constant marginal costs and no externalities. Under these assumptions, there is a simple relationship between the competitive quantity and leaders' choices. Without any further knowledge about the model parameters an observer (such as a regulator or an econometrician) can learn the competitive quantity as soon as the first entrant makes a choice. This inference does not require the observer or even the firms to have correct beliefs about the future arrivals of entrants. 

I show that this result is driven by Stackelberg independence. I define Stackelberg independence as a property of sequential games, where each leader behaves independently of the number of followers. To show the importance of this property, I first prove a limit result. If it is commonly known that the number of followers is going to be large, so that the total quantity will be very close to perfectly competitive quantity, then the leaders' actions are proportional to competitive quantity. This connection holds both with linear and non-linear demand functions. The reason is simple: near the competitive quantity, any demand function can be closely approximated by a linear demand function. 
The same connection between leaders' choices and the competitive quantity continues to hold with a finite number of followers if and only if the leaders' choices are independent of the number of followers, i.e.\ under the Stackelberg independence property. 

The second part of the paper provides cautionary results. It shows that all assumptions of the standard model are necessary for the results to hold. Moreover, I provide an example that shows that even small deviations from the standard model may make the leaders' choices uninformative about market conditions. Therefore, one should be cautious with policy implications from the standard model.

The intuitive reason why the standard model is special is simple. Each potential entrant affects the incentives of the leader in two ways. First, an additional firm increases the total equilibrium quantity. This reduces the equilibrium price and, therefore, the marginal benefit of producing an additional unit. This effect pushes the quantities of all existing firms downwards. Indeed, in a simultaneous choice model (Cournot oligopoly), this is exactly what we see---each additional firm increases the total quantity, but reduces the individual quantities. 
Second, by increasing its quantity, the leader can discourage the follower from choosing a large quantity. This discouragement effect pushes leaders' quantities upwards. It is the reason why leaders typically choose larger quantities than followers in the Stackelberg model. 
The Stackelberg independence property is satisfied in the knife-edge case where the two effects are exactly equal so that additional followers (or even changed beliefs about the followers) neither increase nor decrease leaders' optimal quantities.

The standard model discussed in this paper has been extensively used in the literature. \cite{daughety-90} used a two-period model with some leaders and some followers to study the benefits of concentration. The model has been later used and extended by \cite{anderson-engers-92,pal-sarkar-2001,lafay-2010,julien-musy-saidi-2011,julien-musy-saidi-2012}; and \cite{ino-matsumura-2012} to cover more than two periods and an arbitrary number of firms in each period. In all those papers, the model has the Stackelberg independence property.
I extend this literature by showing two results. First, the implications continue to hold when allowing stochastic arrival processes and arbitrary beliefs about the arrival process. Second, I show that the characterization result is driven by Stackelberg independence.

On the other hand, there is a large literature studying sequential games with non-quadratic payoffs, including \cite{dixit-87,robson-90,linster-93,glazer-hassin-00,morgan-2003} and \cite{hinnosaar-osc-arxiv}, whose characterization results do not exhibit the Stackelberg independence property.\footnote{For a literature review on Stackelberg games, see \cite{julien-handbook}, and for sequential contests, see \cite{konrad-2009-book}.} This means that not all Stackelberg leadership models have the Stackelberg independence property. In this paper, I show that this is not a coincidence---all the assumptions of the standard model are necessary for Stackelberg independence and the simple characterization obtained in the standard model.

Methodologically, the paper builds on recent results from aggregative and sequential games. While the literature on oligopolies is very established, starting from \cite{cournot} for oligopolies with simultaneous choices and \cite{stackelberg} in leadership models, such games with non-linear demand functions have been difficult to handle. In recent years, the novel results on aggregative games\footnote{See \cite{acemoglu-jensen-2013} and \cite{jensen-2017-handbook}.} have been successfully applied to shed new light on oligopolies, for example, by \cite{nocke-schutz-2018}. In this paper, I build on the characterization results from sequential contests from \cite{hinnosaar-osc-arxiv}.

The paper is organized as follows. \Cref{S:linearmodel} studies the standard model. I first show by example with a single leader how the leader's action is informative and independent of the number of followers. I then characterize the equilibria for the general case and discuss the properties of the equilibrium. \Cref{S:stackelberg_independence} shows how these results are driven by Stackelberg independence property. \Cref{S:necessary} shows how each of the assumptions in the standard model is necessary for the results. \Cref{S:discussion} concludes. The proofs are in \cref{A:proofs}.

\section{The Standard Model} \label{S:linearmodel}

There are $n$ firms producing a homogeneous good with constant marginal cost $c \geq 0$. The (inverse) demand function is linear $P(X) = a \left(\oX - X \right)$, where $X = \sum_{i=1}^n x_i$ is the total quantity produced by all firms, $x_i \geq 0$ is the individual quantity of firm $i$, $a>0$, and $\oX>0$ is the market saturation quantity.

Firms are partitioned into $T$ groups that I call periods. The set of firms arriving in period $t$ is denoted by $\prt_t$ and their number $n_t = \#\prt_t$. That is, $\bprt = \left(\prt_1,\dots,\prt_T \right)$ is a partition of $n=\sum_{t=1}^T n_t$ firms. If firm $i \in \prt_t$ arrives before firm $j \in \prt_s$ (i.e.\ $t < s$), then firm $i$ is a \emph{leader} for firm $j$ and correspondingly firm $j$ is a \emph{follower} for firm $i$.
Firm $i$ arriving in period $t$ observes the cumulative quantity of its leaders, i.e.\ firms that arrived prior to period $t$. I denote this cumulative quantity by $X_{t-1} = \sum_{s=1}^{t-1} \sum_{j \in \prt_s} x_j$. Firm $i$ chooses $x_i$ simultaneously with other firms arriving in period $t$.

Let $x_i^*(\bn)$ denote the equilibrium quantity of firm $i$ when the sequence of firms is $\bn=(n_1,\dots,n_T)$. Stackelberg independence is defined as each $x_i^*(\bn)$ being independent on the sequence of followers $\bn^t = (n_{t+1},\dots,n_T)$. Formal definition is as follows.
\begin{definition}[Stackelberg independence]
  The model satisfies Stackelberg independence property if for all sequences $\bn$, all periods $t$, and firms $i \in \prt_t$, for each $\hbn = (\hn_1,\dots,\hn_{\hT})$ such that $\hn_s = n_s$ for all $s \leq t$, the equilibrium quantity $x_i^*(\hbn) = x_i^*(\bn)$.
\end{definition}

One possible sequence of followers is such that there are no followers. Therefore Stackelberg independence requires that each firm always behaves as if there are no followers. For example, if there is a single first-mover (that is $n_1=1$) then Stackelberg independence implies that the first-mover chooses monopoly quantity regardless of the actual sequence of followers.
Note that the property does not put any restrictions on the off-path behavior, nor does it prohibit the equilibrium behavior of firm $i$ depending on the number of firms arriving either at the same period or earlier than firm $i$.

\subsection{Example} \label{SS:linear}

Suppose that the inverse demand is linear $P(X) = a \left( \oX -X \right)$ and the marginal cost is $c \geq 0$. The competitive equilibrium quantity $\oX_c = \oX - \frac{c}{a}$ solves $P(\oX_c) = c$, so that $P(X)-c = a \left(\oX_c-X \right)$.
Firms arrive in two periods. In the first period, a single leader (called firm $1$) arrives and chooses quantity $x_1$. In the second period $n-1\geq 0$ followers arrive, observe $x_1$ and choose their quantities $x_2,\dots,x_n$ simultaneously. The total quantity is $X = \sum_{i=1}^n x_i$.

Straightforward calculations show that in equilibrium the best-response of each follower $i > 1$ is $x_i^*(x_1) = \frac{1}{n} \left( \oX_c - x_1 \right)$. Firm $1$ takes this into account and solves
\begin{equation}
  \max_{x_1} x_1 a \left(
    \oX_c - x_1 - \sum_{i = 2}^n x_i^*(x_1)
  \right)
  =
  \frac{1}{n} \max_{x_1} x_1 a \left(\oX_c - x_1\right).
\end{equation}
Note that $n$ enters the maximization problem multiplicatively. While $n$ affects the leader's profit, it does not affect the maximization problem and the maximizer. 
The leader's optimal quantity is the monopoly quantity $x_1^* = \frac{\oX_c}{2}$ regardless of $n$.

The equilibrium quantities for various values of $n$ are illustrated by \cref{F:linear}. 
As $n$ increases, the total equilibrium quantity (marked with circles) raises to competitive quantity $\oX_c$, but the leader's quantity (solid horizontal line) remains constant at $x_1^*=\frac{\oX_c}{2}$. 
This is an example of the Stackelberg independence property---the leader's behavior is independent of the number of followers.
\begin{figure}[!ht]%
    \centering
    \includegraphics[width=0.45\textwidth,trim={5px 12px 17px 10px},clip]{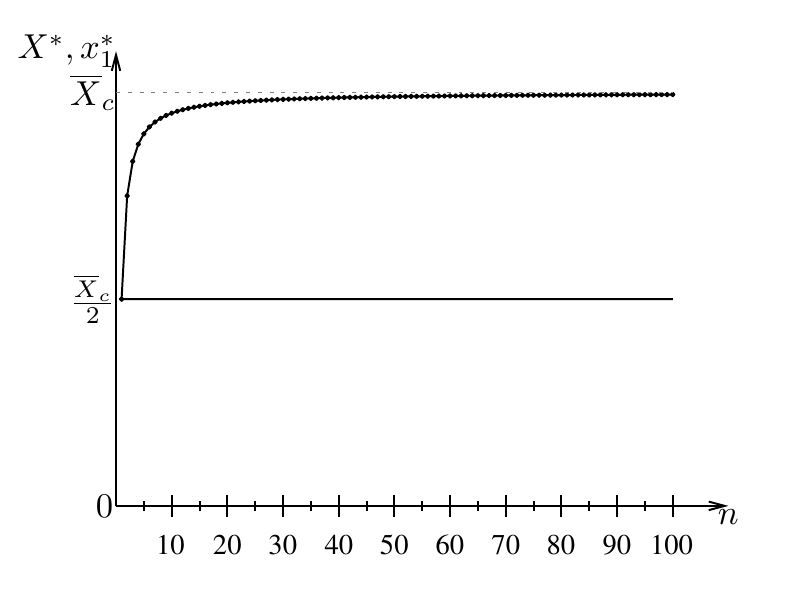}
    \caption{Equilibria with linear demand function in Stackelberg model with one leader and $n-1$ followers. \label{F:linear}}%
\end{figure}

Just by observing the leader's quantity $x_1^*$, an observer can immediately determine the competitive quantity $\oX_c = \oX - \frac{c}{a} = 2 x_1^*$. This does not require the observer to wait until the followers have made their choices. Moreover, the observer does not need to know either the number of followers or what the leader thinks the number of followers is.
When the followers arrive and make their choices, the observer can determine how competitive the market is by comparing the total quantity $X^*$ to the competitive quantity $\oX_c$. The dead-weight loss is proportional to the distance $\oX_c - X^*$.

\subsection{Equilibria} \label{SS:characterization_linear}

I now show that all the properties of the previous example generalize to arbitrary sequence $\bn$. Moreover, to formalize the fact that even the beliefs about the arrival process do not play any role in the equilibrium characterization, I allow a general arrival process and beliefs in this subsection. In particular, I assume that the sequence $\bn$ is a random variable that may come from any distribution. The only restriction that I impose is that no firms arrive after some finite period $T$. Then $\bn=(n_1,\dots,n_T)$ and $\bprt=(\prt_1,\dots,\prt_T)$ denote realizations of the random process.

Firm $i \in \prt_t$ that arrives in period $t$ observes the cumulative quantity $X_{t-1}$ and the number of firms arriving in period $n_t$. It may also get some public or private signals. Using all this information, firm $i$ forms a belief about the future arrival process $\bn^t = (n_{t+1},\dots,n_T)$. Note that these beliefs may be different for different firms in the same period and can depend on $X_{t-1}$ as well as on the individual quantities of the leaders.

\begin{proposition}[Characterization Result for the Standard Model] \label{P:standard}
  The total equilibrium quantity $X^*$ and individual quantities $(x_1^*,\dots,x_n^*)$ are given by
  \begin{equation} \label{E:eq_standard}
    X^* = \left[
    1 - \frac{1}{\prod_{s=1}^T (1+n_s)}
  \right] \oX_c, \;\;\text{ and } x_i^* = \frac{\oX_c}{\prod_{s=1}^t (1+n_s)}, \;\; \forall t \in \{1,\dots,t\}, \forall i \in \prt_t.
  \end{equation}
\end{proposition}

The proof in \cref{A:proofs} generalizes the example in the previous section using mathematical induction. As in the example, the best-response function of the firms in the last period is linear. Therefore both the total quantity $X^*$ and the net demand $P(X^*)-c$ induced by each $X_{T-1}$ are linear functions of $X_{T-1}$. Moreover, straightforward calculations show that net demand is $\frac{1}{1+n_T} a \left( \oX_c - X_{T-1} \right)$. Therefore when firms in period $T-1$ form an expectation about this object, the term that depends on the number of followers is multiplicatively separable from the maximization problem and does not affect the optimum. Standard mathematical induction shows that this is true for all players.

\subsection{Properties} 

As in the example, the general characterization in \cref{P:standard} provides a clear connection between the actions of the leaders and the model parameters. Just by observing the choice $x_i^*$ of one player $i$ in period $t$ and the number of firms in each period up to period $t$, an observer can determine the competitive quantity $\oX_c = \oX - \frac{c}{a} = x_i^* \prod_{s=1}^t (1+n_s)$. By observing all quantities, the observer can determine $\oX_c-X^*$, i.e.\ distance from the competitive equilibrium quantity to equilibrium quantity, which is proportional to the dead-weight loss.  
Note that these observations are independent on the arrival process and beliefs about the arrival process.

These properties of the standard model are a consequence of the Stackelberg independence property. The maximization problem of each firm is independent of the number of followers it has. Therefore it is natural that its choice does not depend either on the number of followers or the belief about their arrival process. In the next section, I show that this connection is even deeper. 
The characterization formulas in \cref{P:standard} hold more generally in the limit with a large number of followers. 
Therefore for the same results to hold for a finite number of players, the choices of the leaders must be the same with a finite and infinite number of players, i.e.\ under Stackelberg independence.

\section{Competitive Limits and Stackelberg Independence} \label{S:stackelberg_independence}

In this section, I relax the standard model by allowing demand to be non-linear. In particular, let $P(X)$ be any strictly decreasing and smooth demand function in $[0,\oX]$, such that the first-order necessary conditions are also sufficient and there is a saturation point $\oX$ such that $P(X)=0$ if and only if $X \geq \oX$. This implies that for each $c \geq 0$ there is a unique competitive equilibrium quantity $\oX_c \in \left[0,\oX \right]$ with $P(\oX_c)=c$. 

\subsection{Competitive Limits}

The first result establishes leaders' behavior in competitive limits, where the number of firms converges to infinity. The first part is intuitive---as the number of firms becomes large, the total equilibrium quantity converges to competitive quantity $\oX_c$. This has been shown earlier in various settings, for example by \cite{robson-90} and \cite{hinnosaar-osc-arxiv}. Note that this aggregate limit is the same regardless of the period in which period the number of firms converges to infinity.

The limiting behavior of individual firms depends on the period in which the number of firms is increased. Naturally, if a firm arrives simultaneously with a large number of firms, they all produce negligible quantities. Similarly, each firm that follows an infinitely large number of leaders also produces a negligible quantity in the competitive limit.

The novel part of the proposition is the limit behavior of the finite number of leaders. The leaders' quantities are uniquely determined by the competitive equilibrium quantity $\oX_c$ and the number of firms in each period up to the arrival of this particular leader. Comparing the limiting quantity of a leader in \cref{P:limits} to the one from the standard model in \cref{P:standard} reveals that the leaders' behavior is identical in both cases. This is natural, as the non-linear function could be closely approximated by a linear function near the competitive equilibrium.

\begin{proposition}[Competitive Limits] \label{P:limits}
Let $\oX_c = P^{-1}(c)$ be the competitive equilibrium quantity with inverse demand $P(X)$ and marginal cost $c \geq 0$. Fix a sequence $\bn=(n_1,\dots,n_T)$ and let us increase $n_t$ in a particular period $t$. Then the limiting total quantity $\lim_{n_t \to \infty} X^*(\bn) = \oX_c$ and individual quantities for each firm $i \in \prt_s$ are
\begin{equation} \label{E:limits}
  \lim_{n_t \to \infty} x_i^*(\bn) = 
  \begin{cases}
    0 & \forall s \geq t, \\
    \frac{\oX_c}{\prod_{k=1}^s (1+n_k)} & \forall s<t.
  \end{cases}
\end{equation}
\end{proposition}

The proof is in \cref{A:proofs}, but to illustrate the argument let me discuss the single leader case studied in \cref{SS:linear}, i.e.\ $\bn = (1,n-1)$ with $n \to \infty$ here. 
The followers observe $x_1$ and maximize $x_i \left(P(X)-c \right)$. Combining their optimality conditions gives an equation that defines total equilibrium quantity $X^*(x_1)$ for any given $x_1$, optimality condition for firm $i$ is $x_i^*  = - \frac{P\left( X^*(x_1) \right)-c}{P'\left( X^*(x_1) \right)} =  g\left( X^*(x_1) \right)$. Adding up the conditions for all followers gives a condition
\begin{equation} \label{E:leaderquantity}
  x_1^*
  = X^*(x_1) - (n-1) g\left( X^*(x_1) \right) 
  ,
\end{equation}
which implicitly defines the aggregate best-response function of all followers. Inserting this into the maximization problem of the leader and taking the optimality conditions gives an equilibrium condition
\begin{equation} \label{E:eq1n}
  X^* = n g\left( X^* \right) - (n-1) g'\left(X^*\right) g\left( X^* \right).
\end{equation}
Clearly, as $n \to \infty$, the left-hand $X^* \to \oX_c$ and $g(X^*) \to g(\oX_c)=0$. Moreover, 
\[
  \lim_{n\to \infty} g'(X^*) = -\frac{[P'(\oX_c)]^2 - [P(\oX_c)-c]P''(\oX_c)}{[P'(\oX_c)]} = -1.
\]
Therefore the limit of \cref{E:eq1n} implies that $\lim_{n \to \infty} n g(X^*) = \frac{\oX_c}{2}$. 
Taking the limit from the leader's equilibrium quantity defined by \cref{E:leaderquantity} gives
\[
  \lim_{n \to \infty} x_1^* = \lim_{n \to \infty} n g(X^*) = \frac{\oX_c}{2}. 
\]

\subsection{Stackelberg Independence}

The combination of the limit result with the Stackelberg independence property gives a precise prediction for the equilibrium behavior of firms. Namely, all firms in periods $s<T$ may potentially have a large number of followers. By Stackelberg independence, their equilibrium behavior must be independent of the number of followers, so their equilibrium quantity must always be equal to the limit found in \cref{P:limits}. This is stated as \cref{C:limit}.
\begin{corollary}[Competitive Limits and Stackelberg Independence] \label{C:limit}
    If the model has the Stackelberg independence property, then for all $s < T$ and all $i \in \prt_s$,
    \begin{equation} \label{E:xiunderSI}
       x_i^* = 
      \frac{\oX_c}{\prod_{k=1}^s (1+n_k)}.
    \end{equation}
\end{corollary}

Note that the behavior of the firms in the last period is not determined by \cref{C:limit}, as they do not have any followers and therefore Stackelberg independence has no implications on their behavior. If we extend the argument slightly, by allowing the possibility that they may also have followers and therefore their behavior is also characterized by \cref{E:xiunderSI}, we can add up all equilibrium quantities and get a unique prediction for the total equilibrium quantity
\[
   X^* = \left[
    1
    -
    \frac{1}{\prod_{k=1}^T (1+n_k)}
  \right] \oX_c.
\]

\section{Necessary Conditions} \label{S:necessary}

In this section, I show that the assumptions of the standard model discussed above are necessary for the results. I relax the assumptions of the standard model one-by-one and show that the Stackelberg independence property fails.

\subsection{Linearity} \label{SS:linearity}

The first main assumption of the standard model is that the demand function is linear. More precisely, the difference between demand and marginal cost is linear, i.e.\ $P(X)-c = a\left(\oX_c-X \right)$. 
Relaxing this assumption, I assume that $P(X)$ is a continuously differentiable (but possibly non-linear) function that satisfies the regularity conditions so that the equilibrium is interior and the second-order conditions for each firm are satisfied.
The following result shows that if the Stackelberg independence property holds at least for two-period deterministic arrival processes $\bn = (n_1,n_2)$ and for all constant marginal costs $c \geq 0$, then the function $P(X)-c$ must be linear in $X$.

\begin{proposition}[Linearity is Necessary] \label{P:necessary_linear}
    Suppose for all $c \geq 0$, all $\bn=(n_1,n_2)$ the model has the Stackelberg independence property. Then $P(X)-c = a \left(\oX_c-X \right)$ for some $a>0, \oX_c > 0$.
\end{proposition}
The proof is in \cref{A:proofs}. Let me illustrate the key ideas of the proof. First, \cref{P:limits} gives a unique prediction of the leaders' total equilibrium quantity, $X_1^* = \frac{n_1 \oX_c}{1+n_1}$ for any $n_1 \geq 1$ and any $c \geq 0$, where $\oX_c = P^{-1}(c)$ is the competitive quantity. 
On the other hand, in $n_1$-player Cournot oligopoly the equilibrium quantity is characterized by the condition $X^* = n_1 g(X^*,c)$, where $g(X,c)=-\frac{P(X)-c}{P'(X)}$. According to the Stackelberg independence assumption, these two quantities must coincide, which gives a condition that relates $n_1,c$ and the demand function through $g(X,c)$ function and $\oX_c$,
\[
  \frac{n_1 \oX_c}{1+n_1} = g\left(\frac{n_1 \oX_c}{1+n_1},c\right).
\]
The rest of the proof uses this condition to identify the shape of the demand function $P(X)$. It shows that at any point $X \in \left(0,\oX \right)$, its derivative $P'(X)=P'(0)$. As $P$ is assumed to be continuously differentiable, it implies that $P$ is indeed linear.

\paragraph{Mathematical Intuition.} Let me also discuss the mathematical intuition of this result more formally. \cite{hinnosaar-osc-arxiv} provides a general characterization for this type of sequential games with non-quadratic payoffs. The total equilibrium quantity $X^*$ is defined by equation
\[
  X^* 
  =\sum_{k=1}^T S_k(\bn) g_k(X^*),
\]
where $S_1(\bn),\dots,S_T(\bn)$ are integers that capture the informativeness of the game $\bn=(n_1,\dots,n_T)$ and $g_1,\dots,g_T$ are defined recursively as 
\[
  g_1(X) = g(X)=-\frac{P(X)-c}{P'(X)},
  \;\;\;
  g_{k+1}(X) = -g_k'(X) g(X), \forall k=1,\dots,T-1.
\]
The $g_k(X^*)$ terms capture the higher-order strategic substitutability of quantities. In particular, $g_1(X^*)$ captures the fact that increasing the quantity slightly reduces the marginal benefit of firms' own quantity. It is weighted by $S_1(\bn)=n$, i.e.\ the number of firms. Next, $g_2(X^*)$ captures the fact that increase in the leader's quantity reduces the marginal benefits for all its followers, therefore discouraging them. This term is multiplied by $S_2(\bn)$, which counts the number of all leader-follower pairs in the model. The other terms capture the same idea, but at a higher order. For example, $g_3(X^*)$ captures the fact if firm $i$ increases its quantity and firm $j$ observes this and responds, then it affects the benefits for all $j$'s followers. Therefore $i$ also has two-step indirect influences. Again, each such term $g_k(X^*)$ is weighted by the total number of $k$-step paths in the model. Each additional follower typically adds influences on all levels.

The equilibrium quantity of firm $i$ arriving in period $t$ is
\[
  x_i^* 
  = \left[ 1 - \sum_{k=1}^{T-t} S_k(\bn^t) g_k'(X^*) \right] g(X^*)
  = g_1(X^*) + \sum_{k=2}^{T+1-t} S_{k-1}(\bn^t) g_k(X^*),
\]
where $S_1(\bn^t),\dots,S_{T-t}(\bn^t)$ capture the informativeness of the remainder game after the move of player $i$, i.e.\ $\bn^t = (n_{t+1},\dots,n_T)$, and $g_1,\dots,g_{T-t}$ are defined as above, with the same interpretation.

Each additional follower increases the total quantity $X^*$, which therefore reduces firm $i$'s incentive to increase quantity. This effect works by reducing $g_k(X^*)$ terms, which are strictly decreasing in the case of higher-order strategic substitutes. However, it also increases informativeness of $\bn^t$, which means that firm $i$ influences more firms. This increases $x_i^*$ directly. Depending on the demand function, the comparison of these two opposite effects can go in either direction. 

In the case of linear demand, these two effects are exactly equal. Namely, if $P(X)=a \left( \oX-X \right)$, then $g(X) 
= \oX_c - X$, where $\oX_c=\oX-\frac{c}{a}>0$, and therefore each $g_k(X) = \oX_c-X$. Combining these, we get that
\[
  X^* = (\oX_c-X^*) \sum_{k=1}^T S_k(\bn)
  \Rightarrow
  \oX_c-X^*
  = \frac{\oX_c}{1 + \sum_{k=1}^T S_k(\bn)}
  = \frac{\oX_c}{\prod_{k=1}^T (1+n_k)}
\]
and
\[
  x_i^*
  = (\oX_c-X^*) \left[1 + \sum_{k=1}^{T-t} S_k(\bn^t) \right]
  = \frac{\oX_c}{\prod_{k=1}^T (1+n_k)} \prod_{k=t+1}^T (1+n_k)
  = \frac{\oX_c}{\prod_{k=1}^t (1+n_k)},
\]
which are the same expressions we derived above and are indeed independent of the number of followers $\bn^t=(n_{t+1},\dots,n_T)$. However, note that the fact that the terms involving $\bn^t$ canceled out in the expression relied on the fact that all the information measures $S_k(\bn)$ had an equal weight $g_k(X^*)=\oX_c-X^*$, which means that all direct and indirect substitutability effects were equal. This is not the case for non-linear demand functions or other deviations from the standard linear model, therefore there is no reason to expect these terms to vanish.

\subsection{Example with a Non-Linear Function} \label{SS:example_nonlinear}

The following example shows that even small deviations from the standard model may make the leaders' actions completely uninformative. 
Consider the same example as in \cref{SS:linear}, but with a small modification in the demand function. Instead of a linear function, let the demand be $P(X) = a \left( \oX - X \right) - \varepsilon \sin(k \pi X)$, where $k \in \N$, $\pi \approx 3.14159$, and $\varepsilon>0$ is sufficiently small so that the regularity conditions are satisfied (the demand is still linear and the equilibrium interior).
By construction, the competitive quantity is still $\oX_c$ and the new demand function differs from the original linear demand function at most by $\varepsilon$. Therefore by taking a small $\varepsilon$, we can closely approximate the original function. On the other hand, by increasing $k>0$, we can increase the first- and second-order derivatives of $P$, which play an important role in equilibrium characterization.

The total equilibrium quantity $X^*$ is given by \cref{E:eq1n} and the leader's quantity by \cref{E:leaderquantity}. To see how non-linearity changes the result, let us consider $k=4$ and two values $\varepsilon = \pm 0.023 a$. This means that we consider two small deviations from the linear demand curve that are still visible on \cref{F:nonlinear_a}.
\begin{figure}[!ht]%
    \centering
    \subfloat[Demand function]{\includegraphics[width=0.45\textwidth,trim={10px 10px 15px 10px},clip]{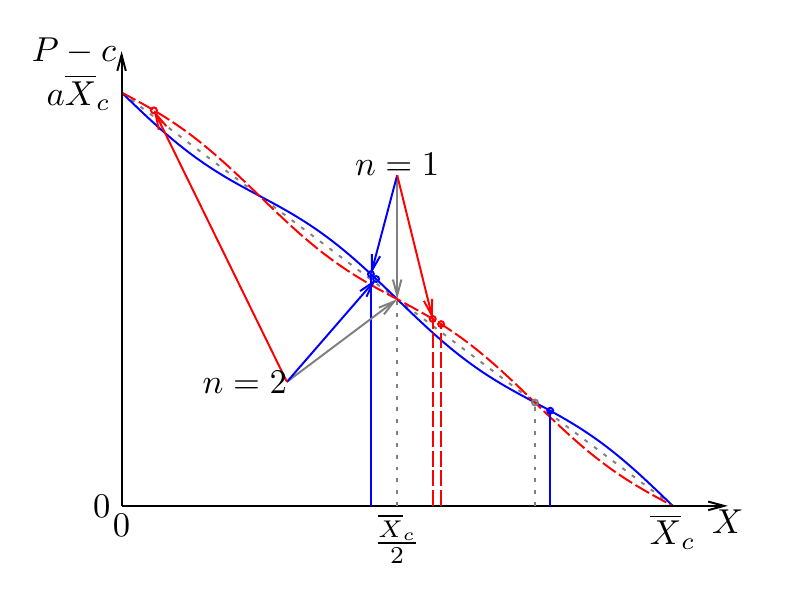} \label{F:nonlinear_a}}
    \qquad
    \subfloat[Equilibrium quantities]{\includegraphics[width=0.45\textwidth,trim={5px 12px 17px 10px},clip]{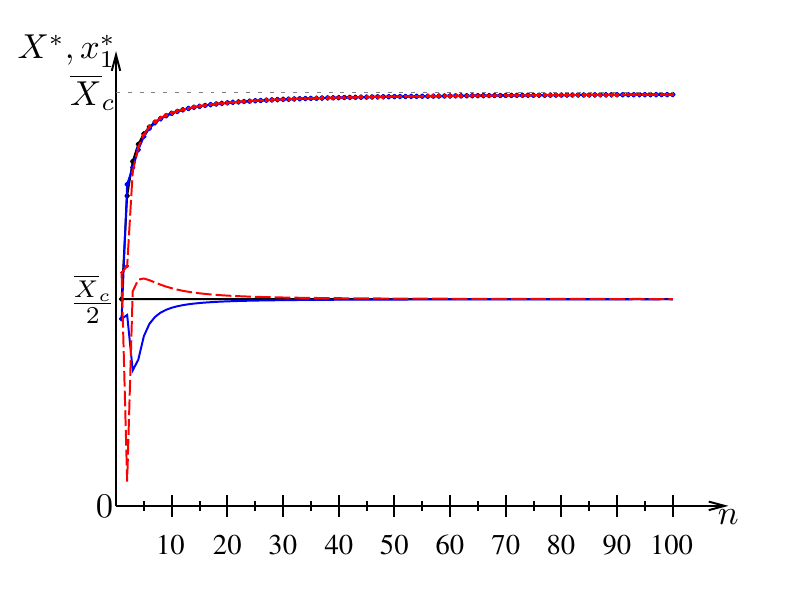} \label{F:nonlinear_b}}%
    \caption{Equilibria with demand functions $P(X)=a\left(\oX-X \right) \pm 0.023 a \sin(5 \pi X)$ in Stackelberg model with one leader and $n-1$ followers. \label{F:nonlinear}}%
\end{figure}

\Cref{F:nonlinear_a} illustrates why the conclusions differ in the case of non-linear demand. First, if $n=1$, i.e.\ the leader is a monopolist, then near the original monopoly quantity $\frac{\oX_c}{2}$, the slopes of the two non-linear demand curves differ. When $\varepsilon < 0$ (the solid blue line) the slope near $\frac{\oX_c}{2}$ is steeper than $-1$, therefore the monopoly quantity is now lower than $0.5$. On the other hand, when $\varepsilon > 0$ (the dashed red line) the curve is less steep than the original demand curve, which pushes the monopoly quantity towards the right. These three monopoly quantities are denoted by vertical lines in the middle of \cref{F:nonlinear_a}. As we can see, relatively small differences in the demand curves lead to visible numerical differences in monopoly quantities.

Next, if there is one follower ($n=2$), then the total equilibrium quantity is closer to the competitive quantity $\oX_c$ (the three vertical lines towards the right). Near the original equilibrium quantity $\frac{3}{4} \oX_c$, the $\varepsilon<0$ case (the solid blue line) has a less elastic demand and thus higher total equilibrium quantity than the linear curve and $\varepsilon>0$ case (the dashed red line) has a more elastic demand and therefore lower total quantity. The leader's corresponding quantities for $\varepsilon<0$ and for $\varepsilon>0$ are now reversed in order and differ significantly.

\Cref{F:nonlinear_b} shows the leader's and the total equilibrium quantity as a function of the number of firms. It shows that the number of followers has a significant impact on the equilibrium behavior of the leader, so we do not have Stackelberg independence here. Moreover, the leader's equilibrium quantity is significantly higher or lower than $\frac{\oX_c}{2}$ even in the case of the same demand function. This means that now we cannot conclude that much from the leader's quantity $x_1^*$. Its connection with the competitive quantity $\oX_c$ depends on the leader's expectation about the number of followers and the particular shape of the demand function. Indeed, suppose that the observer knows that the demand function is one of the two non-linear curves indicated in \cref{F:nonlinear_a}, say with equal probabilities. Then for each $X$, the expected price is $P(X)=a \left(\oX_c-X\right)$, i.e.\ in expectation, the curve is linear, with an error term less than $0.005$. The total equilibrium quantity behaves as one would expect, converging to competitive quantity. However, the leader's quantity $x_1^*$ depends largely on the particular demand function and also the expected number of followers.

Finally, \cref{F:nonlinearsmall} describes the same calculations for $\varepsilon = 0.00025$ and $k=100$. As the figure illustrates, the demand curve is now virtually indistinguishable from the linear curve (having nevertheless sizable derivatives). Vertical lines in \Cref{F:nonlinearsmall_a} show that the leader's action when $n$ is either $1$, $2$, or $3$ differs significantly and \cref{F:nonlinearsmall_b} shows that the same pattern continues for larger $n$ and the convergence to $\frac{\oX_c}{2}$ is slow as $n \to \infty$. Therefore the leader's action can be uninformative and depends heavily on $n$ even when the demand function is very close to linear.
\begin{figure}[!ht]%
    \centering
    \subfloat[Demand function]{\includegraphics[width=0.45\textwidth,trim={10px 10px 15px 10px},clip]{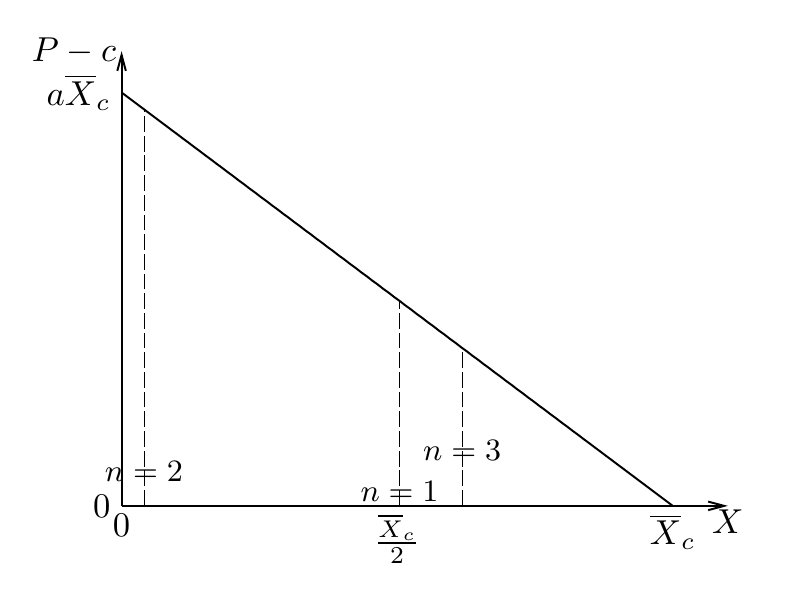} \label{F:nonlinearsmall_a}}%
    \qquad
    \subfloat[Equilibrium quantities]{\includegraphics[width=0.45\textwidth,trim={5px 12px 17px 10px},clip]{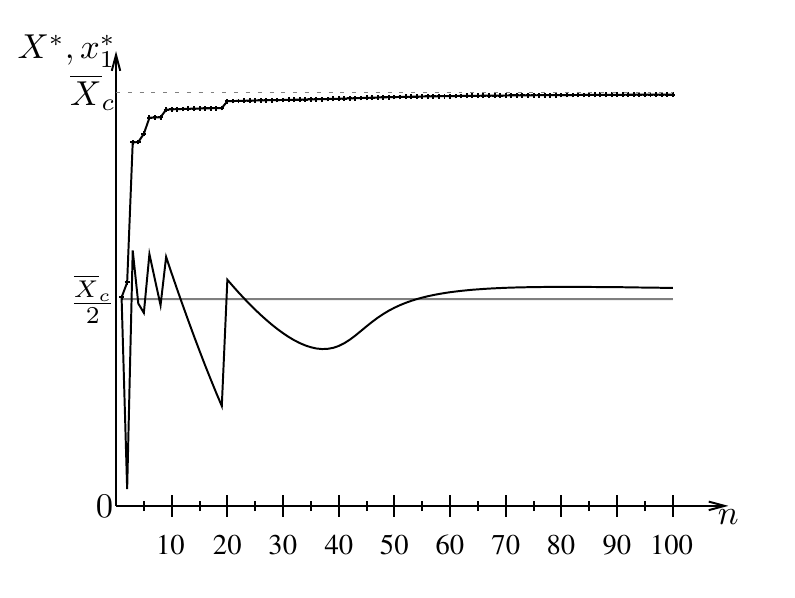} \label{F:nonlinearsmall_b}}%
    \caption{Equilibria with demand function $P(X)=a\left(\oX-X \right) - 0.00025 a \sin(100 \pi X)$ in Stackelberg model with one leader and $n-1$ followers. \label{F:nonlinearsmall}}%
\end{figure}

\subsection{Identical Firms}

The second main assumption of the standard model is that the firms are identical. In this section, I relax this assumption while keeping the other assumptions of the standard model unchanged. I assume that each firm $i$ may have a different constant marginal cost $c_i \geq 0$  and a different linear inverse demand function $P_i(X)=a_i \left(\oX^i - X \right)$. The latter captures the possibility that the firms may operate under different tax rules or have different incentive structures for the decision-makers. 
Under these assumptions, we can rewrite $P_i(X)-c_i = a_i \left(\oX_c^i - X \right)$, where $\oX_c^i = \oX^i - \frac{c_i}{a_i}$ is the total quantity at which firm $i$ would earn zero profit. 
Note that $a_i$ affects the payoff multiplicatively and therefore does not affect the equilibrium behavior. The only relevant parameter is, therefore, $\oX_c^i$. As above, I assume parameters $\oX_c^i$ are commonly known and for simplicity, I focus on the deterministic arrival processes. Moreover, I assume that the differences in $\oX_c^i$ are sufficiently small, so that in equilibrium all firms choose an interior solution to their maximization problem.\footnote{If $
\oX_c^i \gg \oX_c^j$, for example, because $c_j \gg c_i$, then it is to be expected that there could be equilibria where $i$ wants to deter $j$'s entry by choosing a quantity that makes entry unprofitable. In this case, Stackelberg independence is clearly not satisfied, as without $j$'s existence $i$ would choose a lower quantity.}

For a formal statement, I use Stackelberg independence in a specific comparison. The statement requires that for arbitrary sequential oligopoly $\bn$ (with linear payoffs as described above) adding one follower at the end does not change any of the equilibrium choices of previously existing firms. In particular, if the relevant parameter for the added firm is $\oX_c$, then for all its leaders $\oX_c^i = \oX_c$ is necessary for the equilibrium behavior to be unchanged.

\begin{proposition}[Identical Firms are Necessary] \label{P:necessary_identical}
  Suppose that the equilibrium quantities in sequential oligopoly $\bn = (n_1,\dots,n_T)$ with parameters $(\oX_c^i)_{i=1}^n$ are $\bx^* = (x_1^*,\dots,x_n^*)$
  and 
  these quantities remain the same when adding a firm with parameter $\oX_c$ on period $T + 1$. Then $\oX_c^i = \oX_c$ for all firms.
\end{proposition}

The proof is in \cref{A:proofs}.
To illustrate the argument, let us compare a monopoly and a two-firm Stackelberg model. The monopolist maximizes $x_1 a_1 \left( \oX_c^1 - x_1 \right)$ and the monopoly quantity is $x_1^* = \frac{\oX_c^1}{2}$. 
With two sequential oligopolists, the follower's best-response function is $x_2^*(x_1) = \frac{1}{2} \left[ \oX_c - x_1 \right]$. Therefore the leader maximizes
\[
  \max_{x_1} x_1 \left(
    \oX_c^1 - x_1 - x_2^*(x_1)
  \right)
  = \frac{1}{2} \max_{x_1} x_1 \left(
    \oX_c^1-x_1 + \oX_c^1- \oX_c
  \right).
\]
This problem is equivalent to the monopoly problem and gives the same solution as the monopoly problem if and only if the last two terms cancel out, i.e.\ $\oX_c^1 = \oX_c$.

\subsection{No Other Quadratic Payoffs} \label{SS:quadratic}

The remaining assumptions of the standard model are about externalities and non-constant marginal costs. I address these issues by allowing the payoff function of each firm to be a symmetric but more general quadratic function in the following form\footnote{For a discussion about the use and foundations of quadratic games, see \cite{lambert-martini-ostrovsky-2017}.}
\begin{equation} \label{E:quadraticpayoff}
    \pi_i(\bx) 
      = \alpha_0 
      + 
      \alpha_1 x_i
      - \frac{\alpha_2}{2} x_i^2
      + \beta_1 \sum_{j \neq i} x_j
      - \beta_2 \sum_{j \neq i} x_i x_j.
\end{equation}
Let me first give a few comments about this class of functions. First, under the assumption that the equilibrium is interior (i.e.\ satisfying the first-order optimality conditions), the parameter $\alpha_0$ is irrelevant.
Second, this formulation allows quadratic costs, where the marginal cost is either linearly increasing or decreasing. Third, the formulation makes it possible to study oligopolies with differentiated products, where the inverse demand function of firm $i$ is $P_i(\bx)= a (\oX - x_i) - b \sum_{j \neq i} x_j$. When $b=a$, the products are homogeneous (i.e.\ perfect substitutes), if $b < a$ they are imperfect substitutes, and if $b<0$ they are complements. Finally, if $\beta_1 \neq 0$ then there are direct payoff externalities.
The following \cref{P:necessary_quadratic} shows that all these extensions would violate Stackelberg independence.

\begin{proposition}[No Other Quadratic Payoffs] \label{P:necessary_quadratic}
  Suppose that the Stackelberg independence property is non-trivially satisfied for all $\bn=(n_1,n_2)$ and the payoff of player $i$ is given by \cref{E:quadraticpayoff}.
  Then $\alpha_2 = 2 \beta_2$ and $\beta_1=0$ and therefore we can express $\pi_i(\bx)=x_i a \left(\oX_c-X\right)$.
\end{proposition}
\begin{proof}
  For a fixed $X_1$, combining the first-order optimality conditions of all firms in period $2$ gives the total quantity as a function of $X_1$, which is
  \[
    X^*(X_1)
    =\frac{n_2 \alpha_1 + (\alpha_2-\beta_2) X_1}{n_2 \beta_2 + (\alpha_2-\beta_2)}.
  \]
  Inserting this into the optimization problem of the leaders and solving for the equilibrium conditions gives their total equilibrium quantity
  \begin{equation}
    X_1^*
      = \frac{
        \alpha_1
        (\alpha_2-\beta_2)
        n_1
          - \alpha_1 (
              \alpha_2
              - 2 \beta_2
            )
         n_2
        - \beta_1 \beta_2
        n_1 n_2
        }{
      \left(
        \alpha_2
        - \beta_2
        + n_1 \beta_2
      \right) 
       (\alpha_2-\beta_2)
        }  
  \end{equation}
  Non-trivial Stackelberg independence requires that this expression is independent on $n_2$ for all $n_1$, but is not always $0$. The non-triviality requires that $\alpha_1\neq 0$ and $\alpha_2 \neq \beta_2$. Requirement that $\alpha_1 (\alpha_2 - 2 \beta_2) = 0$ implies therefore that $\alpha_2 = 2 \beta_2$. Finally, $\beta_1 \beta_2 = 0$ implies that either $\beta_1 =0$ or $\beta_2=0$, but the previous two observations exclude the possibility that $\beta_2 = 0$.
  Therefore indeed $\beta_1=0$ and $\alpha_2 = 2 \beta_2$. Inserting this into \cref{E:quadraticpayoff} and noting that $\alpha_0$ does not affect interior equilibria gives the representation
    \[ 
      \pi_i(\bx)
      =
      \alpha_1 x_i
      - \beta_2 x_i^2
      - \beta_2 \sum_{j \neq i} x_i x_j
      =
      x_i \beta_2 \left( 
        \frac{\alpha_1}{\beta_2}
        -  X
      \right).
    \]
\end{proof}

\section{Discussion} \label{S:discussion}

This paper studies the standard model of sequential capacity choices. Standard assumptions are often made for tractability and not necessarily because of empirical validity: firms are identical, demand is linear, marginal costs are constant, and there are no externalities. I show that in this standard model, leaders' actions are informative about the markets due to the Stackelberg independence property. Just by observing a single entrant, an observer can deduce the competitive quantity, and it is easy to construct a good welfare measure just by observing the equilibrium quantity choices. Moreover, under the standard assumptions, these arguments are independent of the arrival process and even firms' beliefs about the arrival process.

The second part of the paper bears negative results. Namely, it shows that all the assumptions of the standard model are necessary for the conclusions. Moreover, an example shows that even small deviations from standard assumptions may lead to large changes in behavior, making the leaders' choices uninformative about the market conditions. Therefore, one should be careful with making the standard assumptions just for tractability.

These results highlight that the standard assumptions used in the literature cover only the knife-edge case where different incentives balance out. An additional follower increases the equilibrium quantity and therefore reduces the incentive to choose high quantities for all firms, whereas having more followers motivates leaders to raise their quantities to discourage followers from raising theirs. These effects cancel each other out only when the demand is linear. Similarly, if the follower has a higher or lower cost than the leader or if the goods are non-homogeneous or there are externalities, then these effects do not cancel out even in the case of linear net demand functions.

I did not discuss a few other standard assumptions that can be relaxed and would also affect Stackelberg independence. I maintained the assumption that the equilibrium is interior, which requires that there are no fixed costs (or they are small) and firms do not differ much. However, if the fixed costs are large or if the differences between firms' payoffs are significant, then entry and entry deterrence become strategic questions. Of course, this would make Stackelberg independence even less likely to hold.
Similarly, I assumed that there is common knowledge about firms' payoffs and other model characteristics. Relaxing this is also possible and one implication would be that Stackelberg independence continues to hold with interim-identical firms, i.e.\ firms that may differ in realization, but at the moment of their decision, they expect followers to be similar to them.

\bibliography{toomash-stackelberg}
\bibliographystyle{econometrica}

\appendix

\section{Proofs} \label{A:proofs}

\subsection{Proof of \texorpdfstring{\cref{P:standard}}{proposition 1}}

\begin{proof}
  Firm $i$ in the last period $T$ observes $X_{T-1}$ and knows both $n_T$ and the fact that there are no followers. Therefore its maximization problem is
  \[
    \max_{x_i} x_i a \left( \oX_c - X \right)
    \;\;
    \Rightarrow
    \;\;
    x_i^* = \oX_c - X^*(X_{T-1}),
  \]
  where $X^*(X_{T-1})$ is the total quantity induced by $X_{T-1}$ if all firms in period $T$ behave optimally.
  Combining all the optimality constraints gives us
  \[
    \sum_{i \in \prt_T} x_i^* = X^*(X_{T-1}) - X_{T-1} = n_T \left( \oX_c - X^*(X_{T-1}) \right)
    \;\;
    \iff 
    \;\;
    X^*(X_{T-1})  = \frac{n_T \oX_c + X_{T-1}}{1+n_T}.
  \]
  Now take firm $i$ in period $T-1$. An important object in its maximization problem is $P(X^*(X_{T-1}))-c$, i.e.\ the per-unit realized profit, assuming that after its choice, the cumulative quantity is $X_{T-1}$ and followers behave optimally. Note that since the number of followers is random, firm $i$ takes expectation of this term according to its beliefs, i.e.
  \[
    \expect_i \left[ P(X^*(X_{T-1}))-c \right]
    =
    \expect_i a \left( \oX_c - X^*(X_{T-1})) \right)
    = 
    a \left( \oX_c - X_{T-1} \right)
    \expect_i \frac{1}{1+n_T}.
  \]
  Therefore the expected profit of firm $i$ is $x_i a \left( \oX_c - X_{T-1} \right) \expect_i \frac{1}{1+n_T}$, which is the same problem as if the game would end after period $T-1$.
  
  I prove the proposition by induction. Suppose that at period $t$, each player $i$ expects that cumulative quantity $X_t$ induces 
  $\expect_i \left[ P(X^*(X_t))-c \right]
  = a \left( \oX_c - X_t \right)
    \expect_i \frac{1}{\prod_{s=t+1}^T (1+n_s)}$ (note: we already verified this for $t=T$ and $t=T-1$). Then $i$ maximizes
    \[
      \max_{x_i}  \expect_i x_i \left[ P(X^*(X_t))-c \right]
      = 
      a \expect_i \frac{1}{\prod_{s=t+1}^T (1+n_s)} \max_{x_i} x_i \left( \oX_c - X_t \right).
    \]
    This is clearly independent on $\bn^t$. Combining optimality conditions $x_i^* = \oX_c-X_t^*$ leads to the cumulative equilibrium quantity after period $t$ induced by $X_{t-1}$, which I denote by $X_t^*(X_{t-1})$.
    \[
      \sum_{i \in \prt_t} x_i^* 
      = X_t^*(X_{t-1}) - X_{t-1}
      = \oX_c-X_t^*(X_{t-1})
      \;\;
      \iff
      \;\;
      X_t^*(X_{t-1}) 
      = \frac{n_t \oX_c + X_{t-1}}{1+n_t}.
    \]
    Taking the expectation from the perspective of firm $i \in \prt_{t-1}$ indeed gives
    \[
    \expect_j \left[
      P(X^*(X_t^*(X_{t-1}))-c
    \right]
    = 
    a \left( \oX_c - X_{t-1} \right) \expect_j \frac{1}{\prod_{k=t}^T (1+n_k)}
    .
  \]
  Using these results and the fact that cumulative quantity in the beginning of the game is $X_0=0$, we get that $X_1^*(0) = \frac{n_1 \oX_c}{1+n_1}$, therefore for each $i \in \prt_1$ the equilibrium quantity is $x_i^* = \frac{\oX_c}{1+n_1}$. Then the total quantity at the second period is $X_2^*(X_1^*(0)) - X_1^*(0) =
    \frac{n_2 \oX_c + X_1^*(0)}{1+n_2}
    - X_1^*(0)
    =
    \frac{n_2\oX_c}{(1+n_1)(1+n_2)}
  $ and therefore for each $i \in \prt_2$ we get $x_i^* = \frac{\oX_c}{(1+n_1)(1+n_2)}$.
  By the same argument, for each $i \in \prt_t$, we have $x_i^* = \frac{\oX_c}{\prod_{s=1}^t (1+n_s)}$ and
  \[
    X^* = X_T^*(X_{T-1}^*(\dots(X_1^*(0))\dots)) = \left[
    1 - \frac{1}{\prod_{s=1}^T (1+n_s)}
  \right] \oX_c.
  \]
\end{proof}

\subsection{Proof of \texorpdfstring{\cref{P:limits}}{proposition 2}}
\begin{proof}
  By \cite{hinnosaar-osc-arxiv}, the total equilibrium quantity is characterized by
  \begin{equation} \label{E:eqX_general}
    X^* 
    = \sum_{k=1}^T S_k(\bn) g_k(X^*),
  \end{equation}
  where $g_1,\dots,g_k$ are recursively defined as $g_1(X)=g(X)=-\frac{P(X)-c}{P'(X)}$ and $g_{k+1}(X)=-g_k'(X) g(X)$, and $S_k(\bn)$ denotes the number of level-$k$ observations\footnote{$S_1(\bn)$ is the number of players, $S_2(\bn)$ is the number of players observing other players, etc.} in game $\bn$.
  The equilibrium quantity of firm $i$ arriving in period $s$ is
  \begin{equation} \label{E:eqxi_general}
    x_i^*
    = f_s'(X^*) g(X^*)
    = g(X^*) \left[ 1 - \sum_{k=1}^{T-s} S_k(\bn^s) g_k'(X^*)
  \right]  
  \end{equation}
  \cite{hinnosaar-osc-arxiv} also shows that $\lim_{n_t \to \infty} X^* = \oX_c$. The following \cref{L:glimits} shows that the limits $g_k(\oX_c)=0$ and $g_k'(\oX_c)=-1$ for all $k$.
  \begin{lemma} \label{L:glimits}
    For all $k=1,\dots,T$, $g_k(\oX_c)=0$ and $g_k'(\oX_c)=-1$.
  \end{lemma}
  \begin{proof}
    Clearly $g_1(\oX_c)=g(\oX_c) = -\frac{P(\oX_c)-c}{P'(\oX_c)}=0$ and 
    \[
      g_1'(\oX_c)
      = g'(\oX_c) 
      = -\frac{[P'(\oX_c)]^2-[P(\oX_c)-c]P''(\oX_c)}{[P'(\oX_c)]^2}
      = -1 -g(\oX_c) \frac{P''(\oX_c)}{P'(\oX_c)} = -1.
    \]
    Suppose that $g_k(\oX_c)=0$ and $g_k'(\oX_c)=-1$. Then 
        \begin{align*}
          g_{k+1}(\oX_c)
          &= -g_k'(\oX_c) g(\oX_c) = 0 \\
          g_{k+1}'(\oX_c)
          &= 
          -g_k''(\oX_c) g(\oX_c)
          -
          g_k'(\oX_c) g'(\oX_c)
          = 0-(-1)(-1)=-1.
        \end{align*}
  \end{proof}
  
  Define $\bn_{-t}=(n_1,\dots,n_{t-1},n_{t+1},\dots,n_T)$, i.e.\ the sequence $\bn$ with $n_t$ left out. Note that $S_k(\bn)$ is the number level-$k$ observations in $\bn$, which can be computed by first taking all level-$k$ observations in the subsequence $\bn_{-t}$ and then adding the new observations involving $n_t$, of which there are $n_t$ times $S_{k-1}(\bn_{-t})$.\footnote{For notational convenience, $S_0(\cdot)$ is always $1$ and $S_T(\bn_{-t})=0$ as there cannot be any level-$T$ observations.}
  Taking the limit $n_t \to \infty$ from \cref{E:eqX_general} then leads to
  \begin{equation} \label{E:eqX_general_lim}
    \oX_c
    = \lim_{n_t \to \infty} \sum_{k=1}^T \left[S_k(\bn_{-t})+n_t S_{k-1}(\bn_{-t}) \right] g_k(X^*)
    = \sum_{k=1}^T  S_{k-1}(\bn_{-t}) \lim_{n_t \to \infty} n_t g(X^*),
  \end{equation}
  as $S_k(\bn_{-t})$ and $S_{k-1}(\bn_{-t})$ are independent of $n_t$ and thus finite integers, and $g_k(X^*) = -g_{k-1}'(X^*) g(X^*)$, where $\lim_{n_t \to \infty} g_{k-1}'(X^*)=-1$. 
 
  The next \cref{L:Sproductform} shows that we can rewrite the sum of the measures $S_k$ in a more convenient product form.
  \begin{lemma} \label{L:Sproductform}
    $1+\sum_{k=1}^T S_k(\bn) = \prod_{k=1}^T (1+n_k)$.
  \end{lemma}
  \begin{proof}
    Proof is again by induction. If $T=1$, then $1+\sum_{k=1}^1 S_k(\bn) = 1+S_1(n_1)=1+n_1$. Suppose the claim holds for $T-1$-period games. Then for $T$-period game $\bn=(n_1,\dots,n_T)$
    \[
      \sum_{k=1}^T S_k(\bn)
      = \sum_{k=1}^T [n_T S_{k-1}(\bn_{-T}) + S_k(\bn_{-T})]
      = \sum_{k=1}^T S_k(\bn_{-T})
      + n_T \sum_{k=1}^T S_{k-1}(\bn_{-T}).
    \]
    As $\bn_{-T}$ is a $T-1$-period game, $S_T(\bn_{-T})=0$ and the induction assumption gives us
    \[
      1+\sum_{k=1}^T S_k(\bn_{-T})
      = 1+\sum_{k=1}^{T-1} S_k(\bn_{-T}) 
      = \prod_{k=1}^{T-1}(1+n_k)
    \]
Also, $S_0(\bn_{-T})=1$, so
\[
  \sum_{k=1}^T S_{k-1}(\bn_{-T})
  = 1 + \sum_{k=2}^T S_{k-1}(\bn_{-T})
  = 1 + \sum_{k=1}^{T-1} S_k(\bn_{-T})
  = \prod_{k=1}^{T-1}(1+n_k).
\]
Combining these observations,
\[
  1+\sum_{k=1}^T S_k(\bn)
  = \prod_{k=1}^{T-1}(1+n_k)
  +n_T \prod_{k=1}^{T-1}(1+n_k)
  = 
  \prod_{k=1}^T (1+n_k)
  .
\]
  \end{proof}
Using the representation from \cref{L:Sproductform}, we can rewrite
\[
  \sum_{k=1}^T S_{k-1}(\bn_{-t})
  = 1+\sum_{k=1}^{T-1} S_k(\bn_{-t})
  = \frac{\prod_{k=1}^T (1+n_k)}{(1+n_t)}
\]
Inserting this expression to \cref{E:eqX_general_lim} gives
\begin{equation} \label{E:limntg}
  \lim_{n_t \to \infty} n_t g(X^*)
  = \frac{\oX_c}{\sum_{k=1}^T S_{k-1}(\bn_{-t})}
  = \frac{(1+n_t) \oX_c}{\prod_{k=1}^T (1+n_k)}.
\end{equation}
Take firm $i \in \prt_s$ in period $s$, whose equilibrium quantity is characterized by \cref{E:eqxi_general}. Taking the limit
\begin{equation} \label{E:limx1}
  \lim_{n_t \to \infty} x_i^*
  = g(\oX_c) - \sum_{k=1}^{T-s} g_k'(\oX_c) \lim_{n_t \to \infty} S_k(\bn^s)  g(X^*)
  = \sum_{k=1}^{T-s} \lim_{n_t \to \infty} S_k(\bn^s)  g(X^*).
\end{equation}
There are two cases. If $t \leq s$, then $n_t$ is not included in $\bn^s$, so each $S_k(\bn^s)$ is a finite integer and therefore $S_k(\bn^s) g(X^*)$ converges to $S_k(\bn^s) g(\oX_c) = 0$. Therefore $\lim_{n_t \to \infty} x_i^* = 0$. 
The second case is when $t > s$, which is the case when player $i$ belongs to a finite set of leaders and is followed by an infinite number of followers. Then we can rewrite \cref{E:limx1} as
\[
  \lim_{n_t \to \infty} x_i^*
  = \sum_{k=1}^{T-s} \lim_{n_t \to \infty} [n_t S_{k-1}(\bn_{-t}^s) + S_k(\bn_{-t}^s)]  g(X^*)
  = \sum_{k=1}^{T-s} S_{k-1}(\bn_{-t}^s) \lim_{n_t \to \infty} n_t g(X^*)   
  .
\]
Rewriting $\sum_{k=1}^{T-s} S_{k-1}(\bn_{-t}^s)$ using representation from \cref{L:Sproductform} and inserting limit from \cref{E:limntg} gives
\[
  \lim_{n_t \to \infty} x_i^*
  = \frac{\prod_{k=s+1}^T (1+n_k)}{(1+n_t)} \frac{(1+n_t) \oX_c}{\prod_{k=1}^T (1+n_k)}
  = \frac{\oX_c}{\prod_{k=1}^s (1+n_k)}
  .
\]
\end{proof}

\subsection{Proof of \texorpdfstring{\cref{P:necessary_linear}}{proposition 3}}

Let $\oX_c = P^{-1}(c)$ denote the competitive quantity, so that $P(\oX_c)=c$.
\begin{proof}
  Consider the case when $n_2=0$, i.e.\ there are $n_1 \geq 1$ firms who make a simultaneous choice. Each firm maximizes $\max_{x_i \geq 0} x_i [P(X)-c]$. The equilibrium is defined the first-order conditions for all firms
  \begin{equation*}
  P(X^*)-c
  x_i^* P'(X^*)
  =0
  \;\;\iff\;\;
  x_i^*
  =g(X^*,c),
  \end{equation*}
  where for brevity I denote $g(X,c) = -\frac{P(X)-c}{P'(X)}$. Adding up the conditions for all firms gives
  \begin{equation} \label{E:cournot}
  \sum_{i \in \prt_1} x_i^* 
  = X^*
  =n_1 g(X^*,c).
  \end{equation}
  By Stackelberg independence, the total quantity of the $n_1$ leaders must be the same for any $n_2$. \Cref{C:limit} shows that it must be equal to
  $X_1^* = \frac{n_1 \oX_c}{1+n_1}$. This gives a condition for $c\geq 0, n_1 \in \N$, 
  \begin{equation} \label{E:gequilibrium}
    \frac{n_1\oX_c}{1+n_1} = n_1 g\left(
    \frac{n_1 \oX_c}{1+n_1},c
    \right).
  \end{equation}
  From this, we can determine some properties of $g(X,c)$, which in turn identifies the properties of the demand function $P(X)$. 
  Fix any $X \in \left(0,\frac{\oX}{2} \right)$. By taking $c = P(2X)$ and $n_1$, we can ensure that $\oX_c = 2X$ and therefore \cref{E:gequilibrium} takes the form
  \begin{equation} \label{E:gequilibrium1}
    \frac{1 \oX_c}{1+1} = X = g(X,c).
  \end{equation}
  Now, take any $n_1 \geq 1$ and some $c'$. The total equilibrium quantity in $n_1$-player Cournot model must then be $\frac{n_1}{1+n_1} \oX_{c'} = \frac{n_1}{1+n_1} P^{-1}(c')$. Note that this is a continuous and monotone function of $c'$, which takes value $\frac{n_1}{1+n_1} \oX_c > \frac{\oX_c}{2} = X$ when $c' = c$ and $0$ when $c' \to \infty$. Therefore there exists $c'$ such that the equilibrium quantity is exactly $\frac{n_1}{1+n_1} \oX_{c'} = X$. The equilibrium condition \cref{E:cournot} is $X = n_1 g(X,c')$. Finally, note that by definition,
  \begin{align*}
  g(X,c)  &= -\frac{P(X)-c}{P'(X)} \Rightarrow P'(X) = -\frac{P(X)-c}{g(X,c)}, \\
  g(X,c') &= -\frac{P(X)-c'+c-c}{P'(X)}
  =g(X,c)\left[1-\frac{c'-c}{P(X)-c}\right]
  = X \left[1-\frac{c'-c}{P(X)-c}\right].
  \end{align*}
  Inserting this function and the values $c = P(\oX_c) = P(2X)$ and $c'=P(\oX_{c'}) = P\left( \frac{1+n_1}{n_1} X\right)$ to the equilibrium condition gives
\[
X 
= n_1 g(X,c')
= n_1 X \left[1-\frac{P\left( \frac{1+n_1}{n_1} X\right)-P(2X)}{P(X)-P(2X)}\right],    
\]
which is equivalent to
\begin{equation} \label{E:n1eq}
P(2X)-P(X)
=
n_1 \left[ P\left(X + \frac{X}{n_1}\right)-P(X) \right]
\end{equation}  
Suppose that $n_1 = 2$. Then applying \cref{E:n1eq} on $\frac{2}{3} X$ gives
\[
P\left( X \right)
= \frac{1}{2} P\left( \frac{2}{3} X \right) + \frac{1}{2} P\left( \frac{4}{3} X \right).
\]
Similarly, when $n_1=3$, the applying \cref{E:n1eq} on $\frac{3}{4} X$ gives
\[
P\left( X \right)
= 
\frac{2}{3} P\left( \frac{3}{4} X \right) + \frac{1}{3} P\left( \frac{3}{2} X \right)
\]
Combining the previous two equations we get
\begin{align}
P\left( 2 X \right) - P\left( X \right)
&=
2 \left[
P\left( X \right)
- 
P\left( \frac{1}{2} X \right)
\right]
. \label{E:halfs}
\end{align}
On the other hand, $n_1 \to \infty$ in \cref{E:n1eq} gives
\[
\lim_{n_1 \to \infty} \frac{P\left(X + \frac{X}{n_1}\right)-P(X)}{\frac{X}{n_1}}
= P'(X)
= \frac{1}{X} \left[ P(2X)-P(X) \right].
\]
Combining this with \cref{E:halfs}, we get
\begin{align*}
P'(X)
&= \frac{1}{X} \left[ P\left( 2X \right) - P\left( X \right) \right]
= \frac{2}{X} \left[
P\left( X \right) - P\left( \frac{X}{2} \right)
\right]
= \frac{2^2}{X} \left[
P\left( \frac{X}{2} \right) - P\left( \frac{X}{2^2} \right)	  
\right]
\\
&= \frac{2^k}{X} \left[
P\left( \frac{X}{2^{k-1}} \right) - P\left( \frac{X}{2^k} \right)	  
\right]
= 
2\lim_{k \to \infty}
\frac{P\left(\frac{X}{2^{k-1}} \right)-P(0)}{\frac{X}{2^{k-1}}}
-
\lim_{k \to \infty}
\frac{P\left(\frac{X}{2^k} \right)-P(0)}{\frac{X}{2^k}}
\\
&= 2P'(0) - P'(0) = P'(0).
\end{align*}
That is, for all $X \leq \frac{\oX}{2}$, we must have $P'(X) = P'(0)$, i.e.
$P(X) = P(0) + P'(0) X$. For all $\frac{\oX}{2} < X \leq \oX$ we can apply \cref{E:halfs} at $\frac{X}{2}$ and get
\[
P\left( X \right)
=
P\left( 2 \frac{X}{2} \right)
=P\left(\frac{X}{2}\right) 
+ 2 \left[
P\left(\frac{X}{2}\right) 
-P\left(\frac{X}{4}\right) 
\right]
= P(0) + P'(0) X.
\]
Noting that $0=P(\oX)=P(0)+P'(0)\oX$, so $P(0)=P'(0) \oX$ and denoting $a = -P'(0) > 0$, we get that $P(X) = a \left( \oX-X \right)$ for some $a>0$ for all $X \in [0,\oX]$.
\end{proof}

\subsection{Proof of \texorpdfstring{\cref{P:necessary_identical}}{proposition 4}}

\begin{proof}
  To compare the two sequential oligopolies in the proposition, i.e.\ the original $T$-period oligopoly and the new $T+1$-period oligopoly with an added firm at the end, I consider sequential oligopoly $\bn = (n_1,\dots,n_T,n_{T+1})$, where $n_{T+1} \in \{0,1\}$.
  With slight abuse of notation, I use $X_t^*(X_{t-1})$ to denote the cumulative quantity after period $t$, conditional on cumulative quantity prior to period $t$ being $X_{t-1}$ and $X_t^*$ to denote the cumulative equilibrium quantity on path, i.e.\ $X_t^* = X_t^*(X_{t-1}^*(\dots(X_1^*(0))\dots))$.
  Note that the assumption states that the realized quantities are independent on $n_{T+1}$.

  If $n_{T+1}=1$, i.e.\ if there is a firm $n+1$ who observes $X_T$, then it maximizes $x_{n+1} a_{n+1} \left(\oX_c - X \right)$, which gives an best-response function $X^*(X_T) = \frac{\oX_c + X_T}{2}$. Of course, when $n_{T+1}=0$, we get that $X^*(X_T) = X_T$. These two cases can be combined by 
  \begin{equation} \label{E:nonidentBRlast}
    X^*(X_T) 
    = \frac{n_{T+1} \oX_c + X_T}{1+n_{T+1}}
    = \oX_c - \frac{\oX_c - X_T}{1+n_{T+1}}.
  \end{equation}
  I prove the claim by induction. Fix a period $t \leq T$ and suppose that at each period $s > t$ all players $i \in \prt_s$ have $\oX_c^i = \oX_c$. Moreover, suppose that the best-responses of the followers imply
  \begin{equation} \label{E:nonidentBR}
    X^*(X_t) 
    = X^*(X_T^*(\dots(X_T)\dots)) 
    = \oX_c - \frac{\oX_c - X_t}{\prod_{s=t+1}^{T+1} (1+n_s)}.
  \end{equation}
  Note that these assumptions are satisfied for $t=T$ as (1) whenever a firm arrives at period $T+1$ by assumption it has the parameter $\oX_c$, and (2) \cref{E:nonidentBRlast}.
  For induction step, note that each firm $i \in \prt_t$ maximizes
  \[
    \max_{x_i} x_i a_i \left(
      \oX_c^i - X^*(X_t)
    \right)
    =
    \frac{a_i}{\prod_{s=t+1}^{T+1} (1+n_s)}
    \max_{x_i} x_i \left(
      \oX_c - X_t
      +
      \prod_{s=t+1}^{T+1} (1+n_s) (\oX_c^i - \oX_c)
    \right).
  \]
  The equilibrium behavior requires that
  \[
    x_i^*(X_{t-1})
    =
    \oX_c - X_t^*(X_{t-1})
    +
    \prod_{s=t+1}^{T+1} (1+n_s) (\oX_c^i - \oX_c).
  \]
  In particular, on the equilibrium path, i.e.\ for $X_{t-1}^*$ by assumption we must have that $x_i^*(X_{t-1}^*)$ and $X_t^*(X_{t-1}^*)$ are independent on $n_{T+1}$. This is only true if $\oX_c^i = \oX_c$. This establishes the first induction assumption. Suppose now that $\oX_c^i = \oX_c$ for all $i \in \prt_t$. Then combining the optimality conditions we get
  \[
    X_t^*(X_{t-1})
    = 
    \sum_{i \in \prt_t} x_i^*(X_{t-1}) + X_{t-1}
    =
    n_t \oX_c - n_t X_t^*(X_{t-1}) + X_{t-1}
    = \frac{n_t \oX_c + X_{t-1}}{1+n_t}.
  \]  
  Inserting this to \cref{E:nonidentBR} gives
  \[
    X^*(X_t)
    =
    X^*(X_t^*(X_t))
 = \oX_c - \frac{\oX_c - X_{t-1}}{\prod_{s=t}^{T+1} (1+n_s)}.
  \]
  This proves the second part of the induction assumption and therefore completes the proof.
\end{proof}

\end{document}